\documentclass[11pt,onecolumn,final]{article} 
\usepackage[latin1]{inputenc}
\usepackage{fullpage}
\usepackage{amsfonts}
\usepackage{amsmath}
\usepackage{amssymb}
\usepackage{amsthm}
\usepackage{graphicx}
\usepackage[ruled,longend]{algorithm2e}
\usepackage{stmaryrd}
\usepackage{cite}
\usepackage{enumerate}
\usepackage{url}
\usepackage{multirow}

\usepackage{versions}
\excludeversion{INUTILE}\excludeversion{TEMPORARY}
\excludeversion{JOHANNES}\excludeversion{PETER}\includeversion{LONG}\includeversion{VLONG}
\excludeversion{FORJOURNAL}

\newcommand{\LCP}[0]{\mathsf{LCP}}
\newcommand{\SA}[0]{\mathsf{SA}}
\newcommand{\SF}[1]{S_{#1}}

\theoremstyle{plain}
\newtheorem{definition}{Definition}
\newtheorem{lemma}[definition]{Lemma}

\theoremstyle{remark}

\title{Inducing the LCP-Array}

\author{
  Johannes Fischer\thanks{Computer Science Department, Karlsruhe
    University, \texttt{johannes.fischer@kit.edu}}
}

\date{}

\begin{document}

\maketitle

\begin{abstract}
We show how to modify the linear-time construction algorithm for suffix arrays based on \emph{induced sorting} (Nong et al., DCC'09) such that it computes the array of \emph{longest common prefixes} (LCP-array) as well. Practical tests show that this outperforms recent LCP-array construction algorithms (Gog and Ohlebusch, ALENEX'11).
\end{abstract}

\section{Introduction}
The \emph{suffix array} is an important data structure in text indexing. It is used to solve many tasks in string processing, from exact and inexact string matching to more involved tasks such as data compression, repeat recognition, and text mining. It is also the basic building block for the more complex text index called the \emph{suffix tree}, either indirectly for index construction, or directly when dealing with \emph{compressed} suffix trees. In all of the above applications (possibly apart from exact string matching), the suffix array is accompanied by its sister-array, the array of longest common prefixes (LCP-array for short).

Since their introduction in the early 1990's, much research has been devoted to the fast construction of suffix arrays. Although it is in principle possible to derive the suffix array from the suffix tree, for which linear-time algorithms had already been discovered earlier \cite{weiner73linear}, for reasons of time and space the aim was to construct the suffix array \emph{directly}, without help of the tree. This long line of research (see \cite{puglisi07taxonomy} for a good reference) culminated in three linear-time algorithms \cite{kaerkkaeinen06linear,ko05space,kim05constructing}. However, these algorithms were notorious for being ``linear but not fast'' \cite{antonitio04new}, as they were slower than other non-linear algorithms that had been discovered before and continued to be discovered afterwards.

This un-satisfactory situation (at least for theoretical practioners or practical theoreticians, who want linear-time algorithms to perform faster than super-linear ones) changed substantially when in 2009 a new linear-time algorithm based on \emph{induced sorting} was presented \cite{nong09linear}. A careful implementation of this approach due to Yuta Mori led to one of the fastest known suffix array construction algorithms, often outperforming all other linear or super-linear implementations.

Less emphasis has been put on the efficient construction of the LCP-array. Manber and Myers \cite{manber93suffix} mentioned that it can be constructed along with their method for constructing the suffix array, but their algorithm ran in $O(n\lg n)$ time and performed rather poor in practice. Kasai et al.~\cite{kasai01linear} gave an elegant algorithm for constructing the LCP-array in linear time. A few refinements of this algorithm led to improvements in either space \cite{manzini04two} or in running time \cite{kaerkkaeinen09permuted}. However, these algorithms could not compete with the carefully tuned algorithms for suffix arrays. This led to the odd situation that the rather difficult task of \emph{sorting} suffixes could be solved faster than the seemingly simpler task of computing longest common prefixes.

This situation changed only recently when a theoretically slow $O(n^2)$ but practically fast LCP-array construction algorithm was presented \cite{gog11fast}. Their algorithm exploits properties of the Burrows-Wheeler-Transformation (BWT) of the text, which must be computed before. The authors of \cite{gog11fast} also sketch how their approach yields a linear-time algorithm (for constant alphabets, otherwise it takes $O(n\lg\sigma)$ time).

Driven by the success of the fast linear-time algorithm based on induced sorting \cite{nong09linear}, we show in this paper how it can be adapted such that it also induces the LCP-values (Sect.~\ref{sect:inducing_lcp}). This results in a new linear-time algorithm for constructing LCP-arrays (for integer alphabets). In Sect.~\ref{sect:practical} we show that an ad-hoc implementation of the theoretical ideas leads to a fast practical algorithm that outperforms all other previous algorithms. An additional advantage of our algorithm is that it does not rely on the BWT, and is hence preferable in situations where the BWT is not already present (such as compressed suffix arrays \emph{not} based on the BWT \cite{navarro07compressed}, for example).

Before detailing our theoretical and practical contributions, in Sect.~\ref{sect:preliminaries} we first introduce some notations, and then review the induced sorting algorithm for suffix arrays.

\section{Previous Work and Concepts}
\label{sect:preliminaries}
\subsection{Suffix- and LCP-Arrays}
\label{sec:succ-data-struct}
Let $T=t_1\dots t_n$ be a text consisting of $n$ characters drawn from an ordered alphabet $\Sigma$ of size $\sigma = |\Sigma|$. The substring of $T$ ranging from $i$ to $j$ is denoted by $T_{i..j}$, for $1\le i \le j \le n$. The substring $T_{i..n}$ is called the $i$'th \emph{suffix} of $T$ and is denoted by $S_i$. As usual, for convenience we assume that $T$ ends in a unique character $\$$ which is not present elsewhere in the text, and that $\$<a$ for all $a\in\Sigma$.

The \emph{suffix array} $\SA[1,n]$ of $T$ is a permuation of the integers in $[1,n]$ such that $\SF{\SA[i-1]} <_\mathrm{lex} \SF{\SA[i]}$ for all $1 < i \le n$. In other words, $\SA$ describes the lexicographic order of the suffixes. 

The array $\LCP$ of \emph{longest common prefixes} is based on the suffix array. It holds the lengths of the longest common prefixes of lexicographically adjacent suffixes, in symbols: $\LCP[i] = \max\{\ell>0\mid T_{\SA[i]..\SA[i]+\ell-1 = \SA[i-1]..\SA[i-1]+\ell-1}\}$ for $1 < i \le n$, and $\LCP[0]=0$. 

\subsection{Constructing Suffix Arrays by Induced Sorting}
\label{sect:inducing_sa}
As the basis of our new LCP-array construction algorithm is the \emph{induced sorting} algorithm for constructing suffix arrays \cite{nong09linear}, we explain this latter algorithm in the following. Induced sorting has a venerable history in suffix sorting, see \cite{itoh99efficient, seward00performance, ko05space}. Its basic idea is to sort a certain subset of suffixes, either directly or recursively, and then use this result to \emph{induce} the order of the remaining suffixes. In the rest of this section, we follow the presentation of Okanohara and Sadakane \cite{okanohara09linear}.

\begin{definition}
  For $1 \le i < n$, suffix $\SF{i}$ is said to be \emph{S-type} if $\SF{i} <_\mathrm{lex} \SF{i+1}$, and \emph{L-type} otherwise. The last suffix is defined to be S-type.
\end{definition}

The type of each suffix can be determined in linear time by a right-to-left scan of $T$: first, $\SF{n}$ is declared as S-type. Then, for every $i$ from $n-1$ to $1$, $\SF{i}$ is classified by the following rule:
\vspace{-5ex}
\begin{center}
\item $\SF{i}$ is S-type iff either $t_i < t_{i+1}$, or $t_i = t_{i+1}$ and $\SF{i+1}$ is S-type.
\end{center}

We further say that an S-suffix $\SF{i}$ is of \emph{type S*} iff $\SF{i-1}$ is of type L.

In $\SA$, all suffixes starting with the same character $c\in\Sigma$ form a consecutive interval, called the \emph{$c$-bucket} henceforth. Oberve that in any $c$-bucket, the L-suffixes precede the S-suffixes. Consequently, we can sub-divide buckets into S-type buckets and L-type buckets.

Now the induced sorting algorithm can be explained as follows:

\begin{enumerate}
\item Sort the S*-suffixes. This step will be explained in more detail below.
\item Put the sorted S*-suffixes into their corresponding S-buckets, without changing their order.
\item Induce the order of the L-suffixes by scanning $\SA$ from left to right: for every position $i$ in $\SA$, if $\SF{\SA[i]-1}$ is L-type, write $\SA[i]-1$ to the current head of the L-type $c$-bucket ($c=t_{\SA[i]-1}$), and increase the current head of that bucket by one. Note that this step can only induce ``to the right'' (the current head of the $c$-bucket is larger than $i$).
\item Induce the order of the S-suffixes by scanning $\SA$ from \emph{right to left}: for every position $i$ in $\SA$, if $\SF{\SA[i]-1}$ is S-type, write $\SA[i]-1$ to the current \emph{end} of the S-type $c$-bucket ($c=t_{\SA[i]-1}$), and \emph{de}crease the current end of that bucket by one. Note that this step can only induce ``to the left,'' and might intermingle S-suffixes with S*-suffixes.
\end{enumerate}

It remains to explain how the S*-suffixes are sorted (step 1 above). To this end, we define:

\begin{definition}
  An \emph{S*-substring} is a substring $T_{i..j}$ with $i\ne j$ of $T$ such that both $\SF{i}$ and $\SF{j}$ are S*-type, but no suffix in between $i$ and $j$ is also of type S*.
\end{definition}

Let $R_1,R_2,\dots,R_{n'}$ denote these S*-substrings, and $\sigma'$ be the number of different S*-substrings. We assign a \emph{name} $v_i \in [1,\sigma']$ to any such $R_i$, such that $v_i < v_j$ if $R_i <_\mathrm{lex} R_j$ and $v_i = v_j$ if $R_i =_\mathrm{lex} R_j$. We then construct a new text $T'=v_1\dots v_{n'}$ over the alphabet $[1,\sigma']$, and build the suffix array $\SA'$ of $T'$ by applying the inducing sorting algorithm \emph{recursively} to $T'$ if $\sigma'<n'$ (otherwise there is nothing to sort). The crucial property \cite{nong09linear} to observe here is that the order of the suffixes in $T'$ is the same as the order of the respective S*-suffixes in $T$; hence, $\SA'$ determines the sorting of the S*-suffixes in $T$. Further, as at most every second suffix in $T$ can be of type S*, the complete algorithm has worst-case running time $T(n) = T(n/2)+O(n) = O(n)$, provided that the \emph{naming} of the S*-substrings also takes linear time, which is what we explain next.

The naming of the S*-substrings is similar to the inducing of the S-suffixes in the induced sorting algorithm (steps 2--4 above), with the difference that in step 2 we put the \emph{unsorted} S*-suffixes into their corresponding buckets (hence they are only sorted according to their first character). Steps 3 and 4 work exactly as described above. At the end of step 4, we can assign names to the S*-substrings by comparing adjacent S*-suffixes naively until we find a mismatch or reach their end; this takes overall linear time.

\section{Inducing the LCP-Array}
\label{sect:inducing_lcp}
We now explain how the induced sorting algorithm (Sect.~\ref{sect:inducing_sa}) can be modified to also compute the LCP-array. The basic idea is that whenever we place two S- or L-suffixes $\SF{i-1}$ and $\SF{j-1}$ at adjacent places $k-1$ and $k$ in the final suffix array (steps 3 and 4 in the algorithm), the length of their longest common prefix can be induced from the longest common prefix of the suffixes $\SF{i}$ and $\SF{j}$. As the latter suffixes are exactly those that caused the inducing of $\SF{i-1}$ and $\SF{j-1}$, we already know their LCP-value $\ell$ (by the order in which we fill $\SA$), and can hence set $\LCP[k]$ to $\ell+1$.

\subsection{Basic Algorithm}
We now describe the algorithm in more detail. We augment the steps of the induced sorting algorithm as follows:

\begin{enumerate}[1$'$.]
\item Compute the LCP-values of the S*-suffixes (see Sect.~\ref{sect:recursive_lcp}).
\item Whenever we place an S*-suffix into its S-bucket, we also store its LCP-value at the corresponding position in $\LCP$.
\item Suppose that the inducing step just put suffix $\SF{\SA[i]-1}$ into its L-type $c$-bucket at position $k$. If $\SF{\SA[i]-1}$ is the first suffix in its L-bucket, we set $\LCP[k]$ to $0$. Otherwise, suppose further that in a previous iteration $i'<i$ the inducing step placed suffix $\SF{\SA[i']-1}$ at $k-1$ in the same $c$-bucket. Then if $i'$ and $i$ are in different buckets, the suffixes $\SF{\SA[i]}$ and $\SF{\SA[i']}$ start with different characters, and we set $\LCP[k]$ to $1$, as the suffixes $\SF{\SA[i]-1}$ and $\SF{\SA[i']-1}$ share only a common character $c$ at their beginnings. Otherwise ($i'$ and $i$ are in the same $c'$-bucket), the length $\ell$ of the longest common prefix of the suffixes $\SF{\SA[i]}$ and $\SF{\SA[i']}$ is given by the \emph{minimum} value in $\LCP[i'+1,i]$, all of which are in the same $c'$-bucket and have therefore already been computed in previous iterations. We can hence set $\LCP[k]$ to $\ell+1$. 
\item As in the previous step, suppose that the inducing step just put suffix $\SF{\SA[i]-1}$ into its S-type $c$-bucket at position $k$. Suppose further that in a previous iteration $i'>i$ the inducing step placed suffix $\SF{\SA[i']-1}$ at $k+1$ in the same $c$-bucket (if $k$ is the last position in its S-bucket, we skip the following steps). Then if $i'$ and $i$ are in different buckets, their suffixes start with different characters, and we set $\LCP[k+1]$ to $1$, as the suffixes $\SF{\SA[i]-1}$ and $\SF{\SA[i']-1}$ share only a common character $c$ at their beginnings. Otherwise ($i'$ and $i$ are in the same $c'$-bucket), the length $\ell$ of the longest common prefix of the suffixes $\SF{\SA[i]}$ and $\SF{\SA[i']}$ is given by the \emph{minimum} value in $\LCP[i+1,i']$, all of which are in the same $c'$-bucket and have therefore already been computed. We can hence set $\LCP[k+1]$ to $\ell+1$. 
\end{enumerate}

\subsection{Finding Minima}
\label{sect:rmq}
To find the minimum value in $\LCP[i'+1,i]$ or $\LCP[i+1,i']$ (steps 3$'$ and 4$'$ above), we have several alternatives. The simplest idea is to scan the whole interval from $i'+1$ to $i$; this results in overall $O(n^2)$ running time. A better alternative would be to keep an array $M$ of size $\sigma$, such that the minimum is always given by $M[c]$ if we induce an LCP-value in bucket $c$. To keep $M$ up-to-date, after each step $i$ we first set $M[c]$ to $\LCP[i]$, and further update all other entries in $M$ that are larger than $\LCP[i]$ by $\LCP[i]$; this approach has $O(n\sigma)$ running time. A further refinement of this technique stores the values in $M$ in sorted order and uses binary search on $M$ to find the minima, similar to the stack used by \cite{gog11fast}. This results in overall $O(n\lg\sigma)$ running time.

Yet, we can also update the minima in $O(1)$ amortized running time, as explained next. Let us first focus on the left-to-right scan (step 3$'$); we will comment on the differences to the right-to-left scan (step 4$'$) at the end of this section. Recall that the queries lie within a single bucket (called $c'$), and every bucket is subdivided into an L- and an S-bucket. The idea is to also subdivide the query into an L- and an S-query, and return the minimum of the two. The S-queries are simple to handle: in step 3$'$, only S*-suffixes will be scanned, and these are static. Hence, we can preprocess every S-type bucket with a static data structure for constant-time range minima, using overall linear space \cite[Thm.~1]{fischer10optimal}. The L-queries are more difficult, as elements keep being written to them during the scan. However, these updates occur in a very regular fashion, namely in a left-to-right manner. This makes the problem simpler: we maintain a \emph{Two-Dimensional Min-Heap} \cite[Def.~2]{fischer10optimal} $\mathcal{M}_{c'}$ for each bucket $c'$, which is initially empty (no L-suffixes written so far). When a new L-suffix along with LCP-value $\ell+1$ is written into its $c'$-bucket, we climb up the rightmost path of $\mathcal{M}_{c'}$ until we find an element $x$ whose corresponding array-entry is strictly smaller than $\ell+1$ ($\mathcal{M}_{c'}$ has an artificial root holding LCP-value $-\infty$ which guarantees that such an element always exists). The new element is then added as $x$'s new rightmost leaf; an easy amortized argument shows that this results in overall linear time. Further, $\mathcal{M}_{c'}$ is stored along with a data structure for constant-time \emph{lowest common ancestor queries} (LCAs) which supports dynamic leaf additions in $O(1)$ worst-case time \cite{cole05dynamic}. Then the minimum in any range in the processed portion of the L-bucket can be found in $O(1)$ time \cite[Lemma~2]{fischer10optimal}. \footnote{Note that it is important to use the Two-Dimensional Min-Heap rather than the usual Cartesian Tree for achieving overall linear time, for the following reason: Although the Cartesian Tree also has $O(1)$ amortized update-time for the operation ``append at end;'' it also needs to relink entire subtrees, rather than only inserting new leaves to the rightmost path \cite{gabow84scaling}. For the relink-operation, no constant-time solutions exist for maintaining $O(1)$-LCAs in the tree (not even in an amortized sense); the best solution we are aware of takes $\alpha(\cdot,n)$ update time \cite{harel84fast}, $\alpha(\cdot,\cdot)$ being the inverse Ackermann function.}

In the right-to-left scan (step 4$'$), the roles of the L- and S-buckets are reversed: the L-buckets are static and the S-buckets dynamic. For the former, we already have the range minimum data structures from the left-to-right scan (the 2d-Min-Heaps together with LCA). For the S-buckets, we now build an additional 2d-Min-Heap along with dynamic LCAs; this works because the S-buckets are filled in a strict right-to-left manner.

What we have described in the preceding two paragraphs was actually more general than what we really needed: a solution to the \emph{semi-dynamic range minimum query problem} with constant $O(1)$ query- and amortized $O(1)$ insertion-time, with the restriction that new elements can only be appended at the end (or beginning, respectively) of the array. Our solution might also have interesting applications in other problems. In our setting, though, the problem is slightly more specific: the sizes of the arrays to be prepared for RMQs are known in advance (namely the sizes of the buckets); hence, we can use any of the (more practical) preprocessing-schemes for (static) RMQs in $O(1)$ worst-case time\cite{fischer07new,alstrup04nearest}, and update the respective structures, which are essentially precomputed RMQs over suitably-sized blocks, whenever enough elements have arrived.

\subsection{Computing LCP-Values of S*-suffixes}
\label{sect:recursive_lcp}
This section describes how to compute the LCP-values of the suffixes in the sample set (step 1$'$ above). The recursive call to compute the suffix array $\SA'$ for the text $T'$ (the text formed by the names of the S*-substrings) also yields the LCP-array $\LCP'$ for $T'$. The problem is that these LCP-values refer to characters $v_i$ in the reduced alphabet $[1,\sigma']$, which correspond to S*-substrings $R_i$ in $T$. Hence, we need to ``scale'' every LCP-value in $\LCP'$ by the lengths of the actual S*-substrings that constitute this longest common prefix: a value $\LCP'[k]$ refers to the substring $v_{\SA[k]}\dots v_{\SA[k]+\LCP'[k]-1}$ of $T'$, and actually implies an LCP-value of $\sum_{i=0}^{\LCP[k]-1}|R_{\SA[k]+i}|$ between the corresponding S*-suffixes in $T$.

A naive implementation of this calculation could again result in $O(n^2)$ running time, consider the text $T=\mathtt{abab}\dots \mathtt{ab}$. However, we can make use of the fact that the suffixes of $T'$ appear lexicographically ordered in $T'$: when ``scaling'' $\LCP'[k]$, we know that the first $m = \min(\LCP[k-1],\LCP[k])$ S*-substrings match, and can hence compute the actual LCP-value as
$$
\sum_{i=0}^{\LCP[k]-1}|R_{\SA[k]+i}| = \underbrace{\sum_{i=0}^{m-1}|R_{\SA[k]+i}|}_{\text{already computed}} + \sum_{i=m}^{m-1}|R_{\SA[k]+i}|\ .
$$
This way, by an amortized argument it is easy to see that each character in $T$ contributes to at most 2 additions, resulting in an overall $O(n)$ running time.

It is possible to stop the recursive LCP-calculation at a certain depth and use any other LCP-array construction algorithm on the remaining (sparse) set of sorted suffixes.

\subsection{Computing LCP-values at the L/S-Seam}
There is one subtlety in the above inducing algorithm we have withheld so far, namely that of computing the LCP-values between the last L-suffix and the first S-suffix in a given $c$-bucket (we call this position the \emph{L/S-seam}). More precisely, when reaching an L/S-seam in step 3$'$, we have to re-compute the LCP-value between the first S*-suffix in the $c$-bucket (if it exists) and the last L-suffix in the same $c$-bucket (the one that we just induced), in order to induce correct LCP-values when stepping through the S*-suffixes in subsequent iterations. Likewise, when placing the very first S-suffix in its $c$-bucket in step 4$'$, we need to compute the LCP-value between this induced S-suffix and the largest L-suffix in the same $c$-bucket. (Note that step 4 might place an S-suffix before all S*-suffixes, so we cannot necessarily re-use the LCP-value computed at the L/S-seam in step 3$'$.)

The following lemma shows that the LCP-computation at L/S-seams is particularly easy:

\begin{lemma}
  \label{lemma:ls_seam}
  Let $\SF{i}$ be an L-suffix, $\SF{j}$ an S-suffix, and $t_i = c = t_j$ (the suffixes are in the same $c$-bucket in $\SA$). Further, let $\ell\ge 1$ denote the length of the longest common prefix of $\SF{i}$ and $\SF{j}$. Then
  $$
  T_{i\dots i+\ell - 1} = c^\ell = T_{j\dots j+\ell - 1}\ .
  $$
\end{lemma}
\begin{proof}
  Assume that $t_{i+k} = c' = t_{i+k}$ for some $2\le k < \ell$ and $c'\ne c$. Then if $c'<c$, both $\SF{i}$ and $\SF{j}$ are of type L, and otherwise ($c'>c$), they are both of type S. In any case, this is a contradiction to the assumption that $\SF{i}$ is of type L, and $\SF{j}$ of type S.
\end{proof}

In words, the above lemma states that the longest common prefix at the L/S-seam can only consist of equal characters. Therefore, a \emph{naive} computation of the LCP-values at the L/S-seam is sufficient to achieve overall linear running time: every character $t_i$ contributes at most to the computation at the L/S-seam in the $t_i$-bucket, and not in any other $c$-bucket for $c\ne t_i$.

\section{Experimental Results}
\label{sect:practical}
We implemented the algorithm from the previous section in C and ran several tests on an AMD Athlon 64 processor, running at 2200 MHz with a 512KB L2-cache and 4GB of main memory. The basis of our implementation was Yuta Mori's linear-time C-implementation of the induced-sorting algorithm \cite{nong09linear}, called \textsf{sais-lite} version 2.4.1 (\url{http://sites.google.com/site/yuta256/sais}). We made the following implementation decisions: instead of calculating the LCP-values of the S*-suffixes recursively, we used a sparse variant of the $\Phi$-algorithm \cite{kaerkkaeinen09permuted} immediately on the first level, which calculates the LCP-values of the S*-suffixes in overall linear time. For the inducing step, we used the simple $O(n\sigma)$-variant described in Sect.~\ref{sect:rmq}. The resulting algorithm is called \textsf{inducing} henceforth.

We compared our implementation to the following LCP-array construction algorithms:
\begin{description}
\item \textsf{KLAAP}: the original linear-time method for constructing LCP \cite{kasai01linear}, implemented in a space-saving variant \cite{manzini04two}.
\item $\Phi$: the $\Phi$-algorithm of K\"arkk\"ainen et al.~\cite{kaerkkaeinen09permuted}, which is a clever variant of \textsf{KLAAP} that avoids cache-misses by reorganizing the computations.
\item \textsf{GO}: the hybrid algorithm as described by \cite{gog11fast}. It needs the Burrows-Wheeler-Transformation (BWT) for LCP-array construction, and computes small LCP-values naively, from which the larger LCP-values are deduced.
\item \textsf{GO2}: a semi-external variant of \textsf{GO} \cite{gog11fast}.
\item \textsf{naive}: for a sanity check, we also included the \emph{naive} computation of the LCP-array (step through the suffix array and compare corresponding suffixes naively).
\end{description}
We used the implementations from the \emph{succinct data structures library} (sdsl 0.9.0) \cite{gog11fast} wherever possible. All programs were compiled using the same compiler options (-ffast-math -O9 -funroll-loops -DNDEBUG).

We chose the test suite from \url{http://pizzachili.dcc.uchile.cl/} for evaluation, which is by now a de-facto standard. It includes texts from natural languages (English), biology (dna and proteins), and structured documents (dblp.xml and sources). Because the authors of \cite{gog11fast} point out that the human chromosome 22 from Manzini's corpus (hs) is a particular hard case for some algorithms, it was also included.

The results are shown in Tbl.~\ref{tbl:results}. The first block of columns shows the running times for pure LCP-array construction. For \textsf{KLAAP} and $\Phi$, these times include construction of the inverse suffix- and the $\Phi$-array, respectively, as they are needed for LCP-array computation. For \textsf{GO} and \textsf{GO2}, the times for computing the BWT are \emph{not} included; the reason is that in some cases the BWT is also needed for other purposes, so it might already be in memory. As \textsf{inducing} is inherently coupled with SA-construction \cite{nong09linear}, we could not measure its running times for pure LCP-array construction directly; the figures in column ``\textsf{inducing}'' of Tbl.~\ref{tbl:results} are hence \emph{obtained} by first running the \emph{pure} SA-construction (\textsf{sais-lite}), then the combined LCP- and SA-construction, and finally taking the difference of both running times. Measured this way, \textsf{inducing} takes always less time than all other methods.

\begin{table}
\begin{center}
\begin{tabular}{r|r||c|c|c|c|c|c||c|c|c||c|c|c|}
& & \multicolumn{6}{c||}{pure LCP-array construction} & \multicolumn{2}{c|}{SA} & BWT & \multicolumn{3}{c|}{\textbf{SA+LCP}}\\
& & \rotatebox{90}{\textsf{KLAAP} \cite{kasai01linear}}& \rotatebox{90}{$\Phi$ \cite{kaerkkaeinen09permuted}} & \rotatebox{90}{\textsf{GO} \cite{gog11fast}} & \rotatebox{90}{\textsf{GO2} \cite{gog11fast}} & \rotatebox{90}{\textsf{naive}} & \rotatebox{90}{\textsf{inducing}$^{(*)}$ [this paper]} & \rotatebox{90}{\textsf{divsufsort}} & \rotatebox{90}{\textsf{sais-lite}} & & \rotatebox{90}{\textsf{GO}+BWT+\textsf{divsufsort}} & \rotatebox{90}{\textsf{naive}+\textsf{divsufsort}} & \rotatebox{90}{\textsf{inducing}+\textsf{sais-lite}} \\\hline\hline
\multirow{5}{*}{\rotatebox{90}{20MB}}
& dna      & 7.1 & 6.3 & 4.2 & 6.2 & 3.1 & 2.6 & 4.9 & 6.8 & 2.6 & 11.7 & \textbf{8.0} & 9.4 \\\cline{2-14}
& english  & 6.3 & 5.5 & 9.9 & 12.2&132.8& 2.8 & 4.9 & 6.5 & 2.6 & 17.4 &137.7& \textbf{9.3} \\\cline{2-14}
& dblp.xml & 5.4 & 5.0 & 4.1 & 6.0 & 3.5 & 2.7 & 4.0 & 5.4 & 2.5 & 10.6 & \textbf{7.5} & 8.1 \\\cline{2-14}
& sources  & 5.1 & 5.0 & 4.5 & 6.7 & 3.8 & 2.5 & 3.5 & 5.4 & 2.2 & 10.2 & \textbf{7.3} & 7.9 \\\cline{2-14}
& proteins & 6.0 & 5.6 & 7.6 & 9.8 & 11.4& 2.5 & 5.1 & 7.3 & 2.5 & 15.2 & 16.5& \textbf{9.8} \\\hline\hline
\multicolumn{2}{r||}{hs (33MB)}&10.8 & 12.2& 6.6 &10.0 & 4.3 & 4.4 & 8.2 & 11.0& 4.3 & 19.1 & \textbf{12.5}& 15.4 \\\hline\hline
\multirow{5}{*}{\rotatebox{90}{50MB}}
& dna      & 20.6& 18.0& 10.8& 16.0& 8.3 & 7.1 & 14.1& 18.1& 6.9 & 31.8 & \textbf{22.4} & 25.2 \\\cline{2-14}
& english  & 18.0& 16.0& 21.2&26.5 &193.7& 7.9 & 13.3& 18.1& 6.8 & 41.3 &207.0 & \textbf{26.0} \\\cline{2-14}
& dblp.xml & 15.1& 14.0& 10.7&15.8 & 9.2 & 6.8 & 11.1& 14.2& 6.3 & 28.1 & \textbf{20.3} & 21.0 \\\cline{2-14}
& sources  & 14.4& 13.6& 14.7& 20.2& 17.7& 6.5 & 9.7 & 14.4& 5.8 & 30.2 & 27.4 & \textbf{20.9} \\\cline{2-14}
& proteins & 19.1& 17.1& 15.4& 20.7& 18.6& 7.1 & 15.7& 22.3& 6.7 & 37.8 & 34.3 & \textbf{29.4} \\\hline\hline
\multirow{5}{*}{\rotatebox{90}{100MB}}
& dna      & 47.0 & 41.1 & 22.0 & 32.9 & 17.4 & 16.2 & 32.1 & 39.1& 14.9 & 69.0 & \textbf{49.5} & 55.3 \\\cline{2-14}
& english  & 40.8 & 36.3 & 38.9 & 49.6 & 547.0& 17.5 & 29.8 & 39.8& 14.5 & 83.2 &576.8 & \textbf{57.3} \\\cline{2-14}
& dblp.xml & 32.0 & 30.0 & 21.7 & 31.9 & 19.5 & 14.8 & 24.1 & 29.6& 13.1 & 58.9 & \textbf{43.6} & 44.4 \\\cline{2-14}
& sources  & 30.3 & 28.6 & 28.2 & 38.7 & 109.4& 13.8 & 20.9 & 30.2& 12.2 & 61.3 &130.3 & \textbf{44.0} \\\cline{2-14}
& proteins & 43.5 & 38.6 & 35.9 & 46.7 & 49.3 & 16.2 & 35.3 & 48.7& 14.5 & 85.7 & 84.6 & \textbf{64.9} \\\hline\hline
\multirow{5}{*}{\rotatebox{90}{200MB}}
& dna      & 104.4& 92.7 & 46.1 & 68.5 & 51.0 & 36.3 & 75.9 & 87.6& 32.7 & 154.7 & 126.9 & \textbf{123.9} \\\cline{2-14}
& english  & 90.7 & 80.9 & 82.3 &104.3 &3190.5& 39.4 & 68.9 & 88.8& 31.6 & 182.8 &3259.4 & \textbf{128.2} \\\cline{2-14}
& dblp.xml & 69.2 & 64.6 & 44.1 & 64.3 & 40.4 & 31.1 & 53.2 & 63.6& 27.4 & 124.7 & \textbf{93.6}  & 94.7  \\\cline{2-14}
& sources  & 65.9 & 62.0 & 58.7 & 79.9 & 141.5& 29.5 & 46.4 & 65.3& 26.0 & 131.1 & 187.9 & \textbf{94.8}  \\\cline{2-14}
& proteins & 91.6 & 82.9 & 82.0 & 105.0& 124.2& 35.6 & 76.5 &104.0& 30.8 & 189.3 & 200.7 & \textbf{139.6} \\\hline
\end{tabular}
{\footnotesize $^{(*)}$ As \textsf{inducing} is inherently coupled with SA-construction (\textsf{sais-lite} in our implementation), the running times for pure LCP-array construction were calculated by taking the difference of ``\textsf{inducing}+\textsf{sais-lite}'' and ``\textsf{sais-lite}.''}
\end{center}
\caption{Running times (in seconds) for LCP- and suffix-array construction. The first block of columns shows the running times for pure LCP-array construction (for KLAAP and $\Phi$, these times include construction of the inverse suffix- and the $\Phi$-array, respectively). The second block shows the construction times of those arrays that need to be constructed before LCP: SA (always) and BWT (for \textsf{GO} and \textsf{GO2}). The third block shows the overall running times for computing both SA and LCP for the best possible combinations of algorithms.}
\label{tbl:results}
\end{table}

A fairer comparison of the algorithms is shown in the last three columns of Tbl.~\ref{tbl:results}, where the combined running times for SA- and LCP-array construction are given (for a selection of the best-performing LCP-algorithms). This is because all other methods for LCP-array construction are independent of the method for constructing SA, and can hence be combined with faster SA-construction algorithms. It is by now widely agreed that Yuta Mori's \textsf{divsufsort} in version 2.0.1 is the fastest known such algorithm (\url{http://code.google.com/p/libdivsufsort/}). Hence, for methods \textsf{GO} and \textsf{naive} we give the overall running times combined with \textsf{divsufsort}, whereas for \textsf{inducing} we give the overall running time of \textsf{sais-lite}, adapted to induce LCP-values as well. Further, for \textsf{GO} we also add the times to compute the BWT, as it is needed for LCP-array construction.

Inspecting the results from Tbl.~\ref{tbl:results}, we see that \textsf{inducing}+\textsf{sais-lite} is usually the best possible combination, sometimes outperformed by \textsf{naive}+\textsf{divsufsort}. In fact, the naive algorithm is rather competitive (especially for small inputs up to 50MB), apart from the English text, which consists of long repetitions of the same texts (and hence has large average LCP).

\section{Conclusions and Outlook}
We showed how the LCP-array can be induced along with the suffix array. A rather ad-hoc implementation outperformed all state-of-the-art algorithms. We point out the following potentials for practical improvements: (1) As suffix- and LCP-values are always written to the same place, an interleaved storage of $\SA$ and $\LCP$ could result in fewer cache misses. (2) As the faster \textsf{divsufsort} is also based on induced sorting, incorporating our ideas into that algorithm could result in better overall performance. (3) Computing the LCP-values of the S*-suffixes recursively up to a certain (well-chosen) depth could be faster than just using the $\Phi$-algorithm on level 0, as in our implementation.

\subsection*{Acknowledgments}
We thank Moritz Kobitzsch for help on programming, and Peter Sanders for interesting discussions.

\bibliographystyle{abbrv}
\bibliography{paper}

\begin{thebibliography}{10}

\bibitem{alstrup04nearest}
S.~Alstrup, C.~Gavoille, H.~Kaplan, and T.~Rauhe.
\newblock Nearest common ancestors: A survey and a new algorithm for a
  distributed environment.
\newblock {\em Theory Comput.\ Syst.}, 37:441--456, 2004.

\bibitem{antonitio04new}
Antonitio, P.~J. Ryan, W.~F. Smyth, A.~Turpin, and X.~Yu.
\newblock New suffix array algorithms --- linear but not fast?
\newblock In {\em Proc.\ Fifteenth Australasian Workshop Combinatorial
  Algorithms (AWOCA)}, pages 148--156, 2004.

\bibitem{cole05dynamic}
R.~Cole and R.~Hariharan.
\newblock Dynamic {LCA} queries on trees.
\newblock {\em SIAM J.\ Comput.}, 34(4):894--923, 2005.

\bibitem{fischer10optimal}
J.~Fischer.
\newblock Optimal succinctness for range minimum queries.
\newblock In {\em Proc.\ LATIN}, volume 6034 of {\em LNCS}, pages 158--169.
  Springer, 2010.

\bibitem{fischer07new}
J.~Fischer and V.~Heun.
\newblock A new succinct representation of {RMQ}-information and improvements
  in the enhanced suffix array.
\newblock In {\em Proc.\ ESCAPE}, volume 4614 of {\em LNCS}, pages 459--470.
  Springer, 2007.

\bibitem{gabow84scaling}
H.~N. Gabow, J.~L. Bentley, and R.~E. Tarjan.
\newblock Scaling and related techniques for geometry problems.
\newblock In {\em Proc.\ STOC}, pages 135--143. ACM Press, 1984.

\bibitem{gog11fast}
S.~Gog and E.~Ohlebusch.
\newblock Fast and lightweight {LCP}-array construction algorithms.
\newblock In {\em Proc.\ ALENEX}, pages 25--34. SIAM Press, 2011.

\bibitem{harel84fast}
D.~Harel and R.~E. Tarjan.
\newblock Fast algorithms for finding nearest common ancestors.
\newblock {\em SIAM J.\ Comput.}, 13(2):338--355, 1984.
\newblock See also FOCS'80.

\bibitem{itoh99efficient}
H.~Itoh and H.~Tanaka.
\newblock An efficient method for in memory construction of suffix arrays.
\newblock In {\em Proc.\ SPIRE/CRIWG}, pages 81--88. IEEE Press, 1999.

\bibitem{kaerkkaeinen09permuted}
J.~K\"arkk\"ainen, G.~Manzini, and S.~J. Puglisi.
\newblock Permuted longest-common-prefix array.
\newblock In {\em Proc.\ CPM}, volume 5577 of {\em LNCS}, pages 181--192.
  Springer, 2009.

\bibitem{kaerkkaeinen06linear}
J.~K\"arkk\"ainen, P.~Sanders, and S.~Burkhardt.
\newblock Linear work suffix array construction.
\newblock {\em J.\ ACM}, 53(6):1--19, 2006.

\bibitem{kasai01linear}
T.~Kasai, G.~Lee, H.~Arimura, S.~Arikawa, and K.~Park.
\newblock Linear-time longest-common-prefix computation in suffix arrays and
  its applications.
\newblock In {\em Proc.\ CPM}, volume 2089 of {\em LNCS}, pages 181--192.
  Springer, 2001.

\bibitem{kim05constructing}
D.~K. Kim, J.~S. Sim, H.~Park, and K.~Park.
\newblock Constructing suffix arrays in linear time.
\newblock {\em J.\ Discrete Algorithms}, 3(2--4):126--142, 2005.

\bibitem{ko05space}
P.~Ko and S.~Aluru.
\newblock Space efficient linear time construction of suffix arrays.
\newblock {\em J.\ Discrete Algorithms}, 3(2--4):143--156, 2005.

\bibitem{manber93suffix}
U.~Manber and E.~W. Myers.
\newblock Suffix arrays: A new method for on-line string searches.
\newblock {\em SIAM J.\ Comput.}, 22(5):935--948, 1993.

\bibitem{manzini04two}
G.~Manzini.
\newblock Two space saving tricks for linear time lcp array computation.
\newblock In {\em Proc.\ Scandinavian Workshop on Algorithm Theory (SWAT)},
  volume 3111 of {\em LNCS}, pages 372--383. Springer, 2004.

\bibitem{navarro07compressed}
G.~Navarro and V.~M\"akinen.
\newblock Compressed full-text indexes.
\newblock {\em ACM Computing Surveys}, 39(1):Article No.\ 2, 2007.

\bibitem{nong09linear}
G.~Nong, S.~Zhang, and W.~H. Chan.
\newblock Linear suffix array construction by almost pure indeced-sorting.
\newblock In {\em Proc.\ DCC}, pages 193--202. IEEE Press, 2009.

\bibitem{okanohara09linear}
D.~Okanohara and K.~Sadakane.
\newblock A linear-time burrows-wheeler transform using induced sorting.
\newblock In {\em Proc.\ SPIRE}, volume 5721 of {\em LNCS}, pages 90--101.
  Springer, 2009.

\bibitem{puglisi07taxonomy}
S.~J. Puglisi, W.~F. Smyth, and A.~Turpin.
\newblock A taxonomy of suffix array construction algorithms.
\newblock {\em ACM Computing Surveys}, 39(2), 2007.

\bibitem{seward00performance}
J.~Seward.
\newblock On the performance of {BWT} sorting algorithms.
\newblock In {\em Proc.\ DCC}, pages 173--182. IEEE Press, 2000.

\bibitem{weiner73linear}
P.~Weiner.
\newblock Linear pattern matching algorithms.
\newblock In {\em Proc.\ Annual Symp.\ on Switching and Automata Theory}, pages
  1--11. IEEE Computer Society, 1973.

\end{thebibliography}

\end{document}